\documentclass[12pt]{article}

\usepackage{mathrsfs}
\usepackage{bbm}
\usepackage{color}
\usepackage{fullpage}
\usepackage{amsfonts}
\usepackage{verbatim}
\usepackage{amsmath}
\usepackage{mathrsfs}
\usepackage{graphicx}
\usepackage{subfigure}
\usepackage{amssymb,amsthm,amsmath,bm, pdfpages, epsfig}
\usepackage{float}
\usepackage{verbatim}
\usepackage{kantlipsum}
\usepackage[titletoc,title]{appendix}
\usepackage{enumitem}
\usepackage{fancyhdr}
\usepackage{fancyvrb}
\usepackage{natbib}
\allowdisplaybreaks

\usepackage{multirow}
\newtheorem{Theorem}{Theorem}[section]
\newtheorem{Lemma}{Lemma}[section]
\newtheorem{Proposition}{Proposition}[section]
\newtheorem{Corollary}{Corollary}[section]
\newtheorem{Remark}{Remark}[section]

\numberwithin{equation}{section}

\title{\bf\large A Marked Cox Model for IBNR Claims: Model and Theory}
\author{\normalsize Andrei L. Badescu, ~~X. Sheldon Lin, ~~Dameng Tang\\
\emph{\normalsize Department of Statistical Sciences}\\
\emph{\normalsize University of Toronto}\\
\emph{\normalsize 100 St. George Street, Toronto}\\
\emph{\normalsize Ontario, Canada M5S 3G3}}
\date{}
\numberwithin{equation}{section}
\begin{document}
\maketitle \vspace*{0.5cm}

\begin{abstract}
Incurred but not reported (IBNR) loss reserving is an important issue for Property \& Casualty (P\&C) insurers. The modeling of the claim arrival process, especially its temporal dependence, has not been closely examined in many of the current
loss reserving models. 
In this paper, we propose modeling the claim arrival process together with its reporting delays as a marked Cox process. Our model is versatile in modeling temporal dependence, allowing also for natural interpretations. This paper focuses mainly on the theoretical aspects of the proposed model. We show that the associated reported claim process and
 IBNR claim process are both marked Cox processes with easily convertible intensity functions and marking distributions. The proposed model can also account for fluctuations in the exposure. By an order statistics property, we show that the corresponding discretely observed process preserves all the information about the claim arrival epochs. Finally, we derive closed-form expressions for both the autocorrelation function (ACF) and the distributions of the numbers of reported claims and IBNR claims. Model estimation and its applications are considered in a subsequent paper, \citet{Badescu2015b}. 
\end{abstract}

\newpage

\section{\large Introduction} \label{section: Introduction}

%Loss reserving is important for insurance companies because it can affect its many aspects of operations such as ratemaking, reinsurance decisions, capital allocation, etc. Recently, with the advent of the new supervisory guideline Solvency II, insurers are required to provide not only best estimates of loss reserves, but also have a better grasp of its uncertainty. Among the many categories of reserves, the one for incurred but not reported (IBNR) claims has its uniqueness and is most challenging in its modeling. While one can produce estimate of reported but not settled (RBNS) claims based on case estimate information, IBNR claims are unobserved events and can only be addressed through mathematical models. The number of IBNR claims is an important quantities. It is a building block of the IBNR reserve. It determines the frequency of claims (number of claims per exposure unit). The frequency-severity method usually can increase the stability of reserve estimates especially for the most recent accident years. One can more easily incorporate inflation.

Loss reserving is fundamental for insurance companies because it affects various aspects of the business, such as ratemaking, solvency control and capital allocation. With the advent of Solvency II, insurers are required to not only provide a best estimate of their future liabilities, but also to have a better grasp of their uncertainty. For Property \& Casualty (P\&C) insurance companies, there are two important types of reserves, namely the incurred but not reported (IBNR) reserve and the reported but not settled (RBNS) reserve, which stem respectively from the potential time delays between the claim occurrence and its reporting time or between the reporting time and the settlement time. %While insurance companies can determine the RBNS reserve on a single claim basis by relying on abundant information (e.g., case adjusters' educated estimate), the IBNR reserve can only be calculated on a collective basis as the number of IBNR claims is unknown. This different nature of RBNS and IBNR reserves shows that it might not be appropriate to treat them by the same methodology (Wuthrich and Merz (2008), page 10).
One potential benefit of separately estimating the RBNS reserve and the IBNR reserve is that the adequacy of the case reserves set by claim adjusters can be judged. This is especially important for actuaries to decide whether to include the case reserves for analysis or not (Friedland (2010), page 14).

Although predicting the number of IBNR claims has only been investigated in a relatively few papers (see \citet{Jewell1989}, \citet{Jewell1990}, \citet{Zhao2009}, \citet{ZhaoZhou2010}), the issue is of practical importance for several reasons. The number of IBNR claims can be used to calculate the claim frequency (the total number of reported and IBNR claims per exposure) and thus allows the incorporation of exposure information into loss reserving. In addition, estimating the claim frequency and the severity components separately makes it feasible to explicitly reflect inflation adjustments for the severity (Friedland (2010), page 205) and to stabilize the uncertainty in projecting the ultimate claim amounts, especially for the most recent accident years (Friedland (2010), page 212).

The existing practice to perform the loss reserving task is based on the so-called ``run-off triangle'' which sums up the claim data per combination of accident year and development year. Various deterministic algorithms, such as the chain-ladder (CL) method, the Bornhutter-Ferguson (BF) method and the frequency-severity method, can be chosen to apply to this triangular data (see Friedland (2010) from a practitioners' perspective). Meanwhile, there is a whole array of stochastic models, so called ``macro-level'' models, aiming to interpret these methods and to analyze the uncertainty of their results (\citet{WuthrichMerz2008}, \citet{WuthrichMerz2015}). Due to the limited number of data points contained in the triangular structure, these macro-level models tend to be over-parameterized and thus produce unstable estimates (\citet{Verdonck2009}). Furthermore, they cannot separately estimate the RBNS reserve and the IBNR reserve without some further granulation of the current data (\citet{Schinieper1991}, \citet{LiuVerrall2009}), or the inclusion of new data such as the numbers of reported claims (\citet{Verrall2010}). To resolve these two issues and other shortcomings, a class of ``micro-level'' models has emerged to use policy-level data to depict the development of individual claims. \citet{Norberg1993a}  proposed a marked nonhomogeneous Poisson model and a general mathematical framework for predicting IBNR claims and reserve calculation. The estimation of the model is considered in \citet{Norberg1993b}.
A case study based on Norberg's model using a liability portfolio is presented in \citet{AntonioPlat2014}. Several papers have demonstrated the advantages of this type of micro-level models
over the macro-level models through case studies or simulation experiments (e.g., \citet{JinFrees2013}, \citet{Huangetal2015}). It is noted that the separate estimation of the RBNS reserve and the IBNR reserve is natural in micro-level models because reporting and settlement delays are explicitly modeled.

Naturally, predicting the IBNR reserve requires the modeling of the claim arrival process (e.g., \citet{Jewell1989}, \citet{Norberg1993a}). In actuarial science, a popular model for this purpose is the nonhomogeneous Poisson process, which, to the best of our knowledge,
 has been used in all the micro-level loss reserving models. However, this aspect of  modeling loss reserving may be improved based on several considerations. A nonhomogeneous Poisson process might not be a reasonable approximation for the claim arrival process of a given portfolio when there exists dependence among the individual claim arrivals due to some environmental variations that affect the whole portfolio (\citet{Grandell1991}).
A Cox process is a more appropriate model in such situations (see Paragraph two of Section 2 for more details).
In addition, modeling the claim arrival process as a nonhomogeneous Poisson process implies independence among the numbers of claims from different accident years, which contradicts the calendar year effect exhibited in the run-off triangle (\citet{Holmherg1994}, \citet{Shietal2012}, \citet{WuthrichMerz2015}). Finally, \citet{Mikosch2009} held an empirical study of the nonhomogeneous Poisson model by using the arrival times in the well-publicized Danish fire insurance data. The data variations exhibited by the bursty arrivals are found to be more than those inherent in the nonhomogeneous Poisson process.

The afore-mentioned issues may be dealt with by incorporating a temporal dependence structure into a model for the claim arrival process. In the loss reserving context, this direction has only been taken up in a few papers recently. \citet{Shietal2012} and \citet{MerzWuthrich2013} respectively impose additive and multiplicative structure upon a Bayesian Gaussian copula to model both calendar year and accident year dependence. Although lacking flexibility in modeling dependence, their approaches lead to an analytic formula for the cumulative claim amounts.

In this paper, we propose to model the claim arrival process together with reporting delays as a marked Cox process. The intensity function of the process will be a piecewise stochastic process generated by a hidden Markov model (HMM) with Erlang state-dependent distributions. The proposed model allows for the fluctuation of the exposure over time. As a Cox process, the proposed model shares a similar interpretation as that of Markov-modulated Poisson process (MMPP). On the other hand, our model is different from MMPP because its underlying Markov process is discrete-time based and its piecewise intensity function consists of random variables instead of constants. Under our model assumption, the associated reported claim process and the IBNR claim process are both marked Cox processes with easily convertible intensity functions and marking distributions. We also derive an analytical formula for the number of reported claims and for the number of IBNR claims. Furthermore, the associated discretely observed process of the proposed model is a Pascal-HMM. Using an order statistics property for the proposed model, we show that
this discrete-time model preserves all the information about the claim arrival epochs. The joint distribution of the discretely observed process is a multivariate Pascal mixture, which is known to be extremely flexible in modeling dependence  (\citet{Badescu2015a}). A closed-form expression for the ACF of the discretely observed process is also obtained in steady state.

In a subsequent paper (\citet{Badescu2015b}), we  develop a fitting procedure to estimate all the model parameters including the number of states therein and the corresponding
transition probabilities. The efficiency of the fitting procedure is illustrated through simulation studies. We also fit the model to a real insurance data set, which obtains very good results.

The paper is structured as follows. In Section \ref{section: Proposed model and its properties} we describe the proposed model for the portfolio claim arrival process and discuss how it may be interpreted and justified. The associated reported claim process and IBNR claim process are presented and discussed
in Section \ref{section: Reported claim process and IBNR claim process}. The corresponding
discretely observed processes along with their desirable properties are presented in Section \ref{section: Properties}. An order statistics property of the model is obtained in Section \ref{section: An order statistics property} and is used to explain how the information from the claim arrival process is preserved. The distributions for the number of reported claims and the number of IBNR claims are derived in Section \ref{section: Predicting the number of IBNR claims: a closed-form expression}. Finally, we  provide some concluding remarks along with several directions for future research in Section \ref{conclusion}.

\section{\large  Model set-up and description} \label{section: Proposed model and its properties}

Following the notation of \citet{Norberg1993a}, suppose that the development of a claim until its reporting time is described as a pair of random variables $(T,U)$, where $T$ is the claim arrival epoch and $U$ is its reporting delay. In chronology of their arrival epochs, $\{(T_i,U_i),i=1,2,\cdots,\}$ constitute the claim arrival process of the portfolio along with reporting delays.
We denote the number of claims process by $\{N^a(t),t\ge 0\}$, where $N^a(t)$ is the number of claim up to time $t$. However, at a given valuation date $\tau$, we are only aware of a claim which has occurred so far only if it has been reported. Consequently, one cannot fully observe the claim arrival process up to time $\tau$: $\{N^a(t),0\le t \le\tau\}$. As a result, even though we specify the proposed model based on the claim arrival process, the model estimation is performed through the reported claim process, which will be analyzed in Section \ref{section: Reported claim process and IBNR claim process}.

As stated in \citet{Grandell1991} and also implicitly described in \citet{Norberg1993a}, a Cox process is a natural choice for modeling risk fluctuations exhibited in a portfolio claim arrival process. Depending on the line of business being considered, the stochastic intensity function can be interpreted as variations in an appropriate environment over time. For example, this environment may include weather conditions in automobile insurance. The environmental variation will affect every policy in the portfolio and all policies are independent conditional on the environmental variation. Since the sum of a large number of independent and sparse point processes is approximately a Poisson process (\citet{Grigelionis1963}), \citet{Grandell1991} argues that the portfolio claim arrival process, being unconditional on environmental variation, can be reasonably modeled as a Cox process.

We here propose to model $\{N^a(t),t\ge 0\}$ as a marked Cox process that is described through two components. First, the marks $\{U_1,U_2,\cdots\}$ are independent random variables with common density function $p_U(u)$ and cumulative distribution function $P_U(u)$. Second, the stochastic intensity function $\Lambda(t)$ is a piecewise stochastic process: $\Lambda(t)=\Lambda_l$, for $d_{l-1}\le t<d_l$, $l=1,2,\cdots$ and $d_0=0$. Here $d_l$, $l=1,2,\cdots,$ are pre-determined time points. In practice, data collection cannot be conducted in continuous time and these time points may thus
be interpreted as data collecting times. As a result, we have a continuous stochastic process with discrete observation times.
We assume that $\{\Lambda_1,\Lambda_2,\cdots\}$ is generated by an Erlang hidden Markov model (Erlang-HMM) with the following structure:
\begin{itemize}
  \item \emph{The hidden parameter process} $\{C_1,C_2,\cdots\}$ is a time-homogeneous Markov chain with a finite state space $\{1,2,\ldots,g\}$. Its initial distribution and transition probability matrix are respectively denoted by row vector $\boldsymbol{\pi}_1$ and matrix $\boldsymbol{\Gamma}=(\gamma_{ij})_{g\times g}$, where $\gamma_{ij}=P(C_l=j|C_{l-1}=i)$.
We assume that the Markov chain is irreducible, aperiodic and all the states are positive recurrent.
We denote the existing limiting distribution of the Markov chain by $\boldsymbol{\delta}$. The assumptions are
very natural from a modeling perspective: each state will be revisited infinitely many times over time; the time between two consecutive visits to the same state is irregular; and the mean time is finite.
	
\item \emph{The state-dependent process} $\{\Lambda_1,\Lambda_2,\cdots\}$ is defined such that each $\Lambda_l$ depends only on the current state $C_l$.
Given that $C_l=i$, we assume that $\Lambda_l$ follows an Erlang distribution with shape parameter $m_i$ and scale parameter $\omega_l\theta$ whose density function is given by
      \begin{equation}
        f_{\Lambda_l|C_l=i}(\lambda)=\frac{\lambda^{m_i-1}e^{-\frac{\lambda}{\omega_l\theta}}}{(\omega_l\theta)^{m_i}(m_i-1)!}\triangleq f(\lambda;m_i,\omega_l\theta),\label{state-dependent distr}
      \end{equation}
      where $\omega_l$ represents the risk exposure of the considered portfolio for the $l$th period.
\end{itemize}

If a Poisson process is used to model the portfolio claim arrival process, then the risk exposure is usually incorporated as a multiplicative factor into the intensity function (see e.g., \citet{Norberg1993a} and \citet{Grandell1991}). Due to the scaling property of the class of Erlang distributions, \eqref{state-dependent distr} is indeed a reasonable way to reflect the fluctuation of the risk exposure over time for the claim arrival process.

The mixed Poisson process and the Ammeter process (\citet{Ammeter1948}) are two commonly used classes of Cox processes, both of them being special cases of the proposed model. A mixed Poisson process is a Cox process with intensity process $\Lambda(t)\equiv \Lambda$, a single random variable. Our model reduces to a mixed Poisson process if $d_1=\infty$. An Ammeter process is a Cox process with $\Lambda(t)=\Lambda_l,~d_{l-1}\le t<d_l$, where $\{\Lambda_1,\Lambda_2,\cdots\}$ is a sequence of i.i.d. random variables. When the initial distribution $\boldsymbol{\pi}_1$ is the limiting distribution $\boldsymbol{\delta}$ and each row of $\boldsymbol{\Gamma}$ is $\boldsymbol{\delta}$, our model degenerates to an Ammeter process.

With the above model specifications, one can easily find the density function of each $\Lambda_l$.
\begin{Proposition}
  For the piecewise stochastic intensity function following the Erlang-HMM structure described above, $\Lambda_l$ is an Erlang mixture with density function
  \begin{equation}
    f_{\Lambda_l}(\lambda)=\sum_{i=1}^g\pi_{li}f(\lambda;m_i,\omega_l\theta),~l=1,2,\ldots,
  \end{equation}
  where $\pi_{li}=P(C_l=i)$ is the $i$th element of the row vector $\boldsymbol{\pi}_l=\boldsymbol{\pi}_1\boldsymbol{\Gamma}^{l-1}$.
  \label{Lambda}
\end{Proposition}
\begin{proof}
  It follows immediately from
the law of total probability and the Chapman-Kolmogorov equation.
\end{proof}

\begin{Remark}
  Since the class of Erlang mixtures is dense in the space of positive continuous distributions (\citet{Tijms1994}), one might think of extending the state-dependent density in \eqref{state-dependent distr} to that of mixed Erlang. However, the resulting $\Lambda_l$ is still Erlang mixture distributed (\citet{kpw13}), which implies that $\eqref{state-dependent distr}$ is sufficient for constructing a flexible model of the piecewise stochastic intensity function.
\end{Remark}

\section{\large Reported claim process and IBNR claim process} \label{section: Reported claim process and IBNR claim process}

As stated in Section \ref{section: Proposed model and its properties}, the proposed model can only be estimated through its associated reported claim process. The IBNR claim process plays a key role in predicting the number of IBNR claims. In this section, we show that both processes are still marked Cox processes and we identify their intensity functions. Although both the reported claim process and the IBNR claim process should be indexed with their corresponding valuation date $\tau$, for notational simplicity we drop this index hereafter.

We denote the reported claim process with respect to valuation date $\tau$ as $\{N^r(t),\allowbreak0\le t\le\tau\}$. This process comprises of those marked points from $\{(T_i,U_i),~i=1,2,\ldots\}$ which satisfies the condition $T_i+U_i\le\tau$. The total number of reported claims is then $N^r(\tau)$. When ordered in chronology of their arrival epochs, the selected marked points are denoted by $\{(T_i^r,U_i^r),~i=1,2,\cdots, N^r(\tau)\}$. In a similar way, we can define the IBNR claim process with respect to valuation date $\tau$, which is denoted by $\{N^{IBNR}(t),0\le t\le\tau\}$ and consists of marked points $\{(T_i^{IBNR},U_i^{IBNR}),i=1,2,\cdots,N^{IBNR}(\tau)\}$, where $N^{IBNR}(\tau)$ is the total number of IBNR claims. It is noted that $\{N^r(t), 0\le t\le\tau\}$ is observable while $\{N^{IBNR}(t), 0\le t\le\tau\}$ is not.

In the following, we first prove that the marked Cox processes are closed under thinning. More discussions about the thinning operation of point processes can be found in \citet{Grandell1997}. For this, we will mainly use the tool of Laplace functional transform (LFT). For a point process $N=\{X_i,~i=1,2,\ldots\}$, its LFT is given by
\begin{equation*}
  L_N(f(x))=E(e^{-\sum_if(X_i)}),
\end{equation*}
where $f(x)$ is a nonnegative function of the point process. Just like the fact that the Laplace transforms have a one-to-one correspondence with random variables, the mapping between LFTs and point processes is also one-to-one. For more properties about LFT, see \citet{Mikosch2009}.

\begin{Theorem}
  Assume that $\bar{N}$ is a marked Cox process on $[0,\infty)$ with a marking space of $R^d$. Its intensity function is $\Lambda(t)$ and its marks $\boldsymbol{Z}_i,~i=1,2,\ldots,$ are independent but position-dependent with density function $p_{\boldsymbol{Z}|t}(\boldsymbol{z})$. Now consider the following thinning probabilities:
\begin{equation}
  p(t,\boldsymbol{z})=\left\{
  \begin{array}{rl}
  1 & \text{if } (t,\boldsymbol{z})\in D,\\
  0 & \text{if } (t,\boldsymbol{z})\notin D,
  \end{array} \right.
  \label{thinning prob}
\end{equation}
where $D$ is a subset of $[0,\infty)\times R^d$. Then the resulting thinned point process $\bar{N}_p$ is still a marked Cox process with intensity function $\Lambda(t)P(Z\in D_t)I_{\{t\in T_D\}}$ and independent yet position-dependent marks having density function $\frac{p_{\boldsymbol{Z}|t}(\boldsymbol{z})}{P(\boldsymbol{Z}\in D_t)}I_{\{\boldsymbol{z}\in D_t\}}$, where $D_t=\{\boldsymbol{z}\in R^d|(t,\boldsymbol{z})\in D\}$ and $T_D=\{t\in[0,\infty)|\exists \boldsymbol{z}~s.t.~(t,\boldsymbol{z})\in D\}$.~
\label{thinning of Cox process}
\end{Theorem}

\begin{proof}
  By Equation (1.33) in \citet{Karr1991}, for any nonnegative function $f(t,\boldsymbol{z})$ on $[0,\infty)\times R^d$, the LFT of $\bar{N}_p$ equals $L_{\bar{N}}\left(-\log\left(1-p(t,\boldsymbol{z})+p(t,\boldsymbol{z})e^{-f(t,\boldsymbol{z})}\right)\right)$. Combining the results in Examples 1.16 and  1.28 in \citet{Karr1991}, this can be calculated as $E\left(e^{-\int_0^\infty\left(1-\int_{R^d}(1-p(t,\boldsymbol{z})+p(t,\boldsymbol{z})e^{-f(t,\boldsymbol{z})})p_{\boldsymbol{Z}|t}(\boldsymbol{z})d\boldsymbol{z}\right)\Lambda(t)dt}\right)$.

  Since the thinning probabilities are given in \eqref{thinning prob}, the previous equation can be further simplified as follows:
\begin{align*}
  & E\left(e^{-\int_{T_D}\left(1-\int_{D_t}\left(1-p(t,\boldsymbol{z})+p(t,\boldsymbol{z})e^{-f(t,\boldsymbol{z})}\right)p_{\boldsymbol{Z}|t}(\boldsymbol{z})d\boldsymbol{z}-
  \int_{D_t^c}\left(1-p(t,\boldsymbol{z})+p(t,\boldsymbol{z})e^{-f(t,\boldsymbol{z})}\right)p_{\boldsymbol{Z}|t}(\boldsymbol{z})d\boldsymbol{z}\right)\Lambda(t)dt}\right)\\
  =& E\left(e^{-\int_{T_D}\left(1-\int_{D_t}e^{-f(t,\boldsymbol{z})}p_{\boldsymbol{Z}|t}(\boldsymbol{z})d\boldsymbol{z}-\int_{D_t^c}p_{\boldsymbol{Z}|t}(\boldsymbol{z})d\boldsymbol{z}\right)\Lambda(t)dt}\right)\\
  =& E\left(e^{-\int_{T_D}\left(\int_{D_t}p_{\boldsymbol{Z}|t}(\boldsymbol{z})d\boldsymbol{z}-\int_{D_t}e^{-f(t,\boldsymbol{z})}p_{\boldsymbol{Z}|t}(\boldsymbol{z})d\boldsymbol{z}\right)\Lambda(t)dt}\right)\\
  =& E\left(e^{-\int_{T_D}\left(P(\boldsymbol{Z}\in D_t)-\int_{D_t}e^{-f(t,\boldsymbol{z})}\frac{p_{\boldsymbol{Z}|t}(\boldsymbol{z})}{P(\boldsymbol{Z}\in D_t)}d\boldsymbol{z}P(\boldsymbol{Z}\in D_t)\right)\Lambda(t)dt}\right)\\
  =& E\left(e^{-\int_{T_D}\left(1-\int_{D_t}e^{-f(t,\boldsymbol{z})}\frac{p_{\boldsymbol{Z}|t}(\boldsymbol{z})}{P(\boldsymbol{Z}\in D_t)}d\boldsymbol{z}\right)\Lambda(t)P(\boldsymbol{Z}\in D_t)dt}\right)\\
  =& E\left(e^{-\int_0^\infty\left(1-\int_{R^d}e^{-f(t,\boldsymbol{z})}\frac{p_{\boldsymbol{Z}|t}(\boldsymbol{z})}{P(\boldsymbol{Z}\in D_t)}I_{\{\boldsymbol{z}\in D_t\}}d\boldsymbol{z}\right)\Lambda(t)P(\boldsymbol{Z}\in D_t)I_{\{t\in T_D\}}dt}\right).
\end{align*}
Again it follows from Examples 1.16 and  1.28  in \citet{Karr1991} that $\bar{N}_p$ is a marked Cox process with intensity function $\Lambda(t)P(Z\in D_t)I_{\{t\in T_D\}}$ and position-dependent marks having density functions $\frac{p_{Z|t}(z)}{P(Z\in D_t)}I_{\{z\in D_t\}}$.
\end{proof}

We now show that if the claim arrival process $\{N^a(t),t\ge0\}$ is a marked Cox process, then both the reported claim process $\{N^r(t),0\le t\le\tau\}$ and IBNR claim process $\{N^{IBNR}(t),0\le t\le\tau\}$ are still marked Cox processes with easily convertible stochastic intensity functions and mark densities.

\begin{Theorem}
  Assume that the claim arrival process $\{N^a(t),t\ge0\}$ is a marked Cox process with stochastic intensity function $\Lambda(t)$ and independent marks $\{U_i,i=1,2,\ldots\}$ following common density function $p_{U}(u)$. Then for a given valuation date $\tau$, its associated reported claim process $\{N^r(t),0\le t\le\tau\}$ and IBNR claim process $\{N^{IBNR}(t),0\le t\le\tau\}$ are also marked Cox processes. Their adjusted stochastic intensity functions are $\Lambda^r(t)=\Lambda(t)P_U(\tau-t)I_{\{0\le t\le\tau\}}$ and $\Lambda^{IBNR}(t)=\Lambda(t)(1-P_U(\tau-t))I_{\{0\le t\le\tau\}}$, respectively, and their independent marks follow adjusted position-dependent mark density functions $p^r_{U|t}(u)=\frac{p_U(u)}{P_U(\tau-t)}I_{\{0\le u\le\tau-t\}}$ and $p^{IBNR}_{U|t}(u)=\frac{p_U(u)}{1-P_U(\tau-t)}I_{\{u\ge\tau-t\}}$, respectively.
  \label{thinning_reported_IBNR}
\end{Theorem}

\begin{proof}
  For a given valuation date $\tau$, the reported claim process $\{N^r(t),0\le t\le\tau\}$ is a thinned point process of $\{N^a(t),t\ge0\}$ by thinning probabilities
  \begin{equation*}
  p(t,u)=\left\{
  \begin{array}{rl}
  1 & \text{if } t+u\le\tau ,\\
  0 & \text{if } t+u>\tau.
  \end{array} \right.
\end{equation*}
By Theorem \ref{thinning of Cox process}, $\{N^r(t),0\le t\le\tau\}$ is again a marked Cox process, whose adjusted stochastic intensity and mark density are as given. For the IBNR claim process $\{N^{IBNR}(t),0\le t\le\tau\}$, the argument is similar except that one chooses $D=\{(t,u)|t+u>\tau\}$.
\end{proof}

Theorem \ref{thinning_reported_IBNR} shows that not only are
both the reported claim process and the IBNR claim process marked Cox processes, but also that their intensity processes and mark densities have intuitive interpretations. For example, the intensity for the reported claim process $\Lambda^r(t)$ is the original intensity $\Lambda(t)$ times the probability that the claim has been reported by time $\tau$. The mark density $p_{U|t}^r(u)$ is the original density $p_U(u)$ conditional on the fact that the claim has been reported by time $\tau$.

\section{\large The corresponding discretely observed processes} \label{section: Properties}

Data is usually aggregated in some discrete form before analysis in practice.
As described in Section \ref{section: Proposed model and its properties}, $d_l; l=0,1,\cdots,$ are interpreted as data collecting times. Hence, the number of observable
claims are $\{N_1,N_2,\cdots\}$, where $N_l$ is the number of claims that arrived during $[d_{l-1},d_l)$, no matter reported or not.
The argument can be extended to the reported claims and
the IBNR claims. In this section we characterize the three
discretely observed processes, which not only are important in their own right, but also play a critical role
in fitting the model to data.

We begin with  the discretely observed claim process.
The following proposition shows that under the model assumptions, this discretely observed process follows a Pascal-HMM.

\begin{Proposition}
  With the proposed Cox process $\{N^a(t),t\ge0\}$, its discretely observed process $\{N_1,N_2,\ldots\}$ follows a Pascal-HMM. Its hidden parameter process is given by $\{C_1,C_2,\ldots\}$. For $l=1,2,\ldots;~i=1,2,\ldots,g$, its state-dependent distribution is Pascal with
  \begin{equation}
    P(N_l=n|C_l=i)=p\left(n;m_i,(d_l-d_{l-1})\omega_l\theta\right),
  \end{equation}
  where
  \begin{equation}
    p(n;m,\theta)=\dbinom{n+m-1}{m-1}\left(\frac{1}{1+\theta}\right)^{m}
    \left(\frac{\theta}{1+\theta}\right)^n. \qedhere
  \end{equation}
\label{Pascal-HMM}
\end{Proposition}

\begin{proof}
  Obviously the hidden parameter process is kept the same and we only need to check that the state-dependent distributions are Pascal. Using the law of total probability, we obtain
  \begin{align*}
    P(N_l=n|C_l=i)=& \int_0^\infty P(N_l=n|\Lambda_l=\lambda,C_l=i)f_{\Lambda_l|C_l=i}(\lambda)d\lambda\\
    =& \int_0^\infty \frac{\left((d_l-d_{l-1})\lambda\right)^ne^{-(d_l-d_{l-1})\lambda}}{n!}f(\lambda;m_i,\omega_l\theta)d\lambda\\
    =& p(n;m_i,(d_l-d_{l-1})\omega_l\theta).
  \end{align*}
  Thus $\{N_1,N_2,\ldots\}$ follows the given Pascal-HMM.
\end{proof}

Without loss of generality, we assume that $\tau=d_k$. We denote the discretely observed reported claim process as $\{N_1^r,\ldots,N_k^r\}$, where $N_l^r$ is the number of claims arrived during $[d_{l-1},d_l)$ and reported by time $\tau$. Similarly, $\{N_1^{IBNR},\ldots,\allowbreak N_k^{IBNR}\}$ denotes the discretely observed IBNR claim process, where $N_l^{IBNR}$ is the number of the claims arrived during $[d_{l-1},d_l)$ but have not yet reported yet at time $\tau$. Using similar arguments in the derivations
of Proposition \ref{Pascal-HMM} and Theorem \ref{thinning_reported_IBNR}, it is straightforward to show that both discretely observed processes are also
Pascal-HMMs.

\begin{Corollary}
  For the proposed claim arrival process, the discretely observed processes of the reported claim process and the IBNR claim process are both Pascal-HMMs. They share the same hidden parameter process $\{C_1,\ldots,C_k\}$ and their state-dependent distributions are
\begin{equation*}
  P(N_l^r=n|C_l=i)=p\left(n;m_i,\left(\int_{d_{l-1}}^{d_l}P_U(\tau-t)dt\right)\omega_l\theta\right)
\end{equation*}
and
\begin{equation*}
  P(N_l^{IBNR}=n|C_l=i)=p\left(n;m_i,\left(\int_{d_{l-1}}^{d_l}\left(1-P_U(\tau-t)\right)dt\right)\omega_l\theta\right),
\end{equation*}
where $l=1,2,\cdots,k; i=1,2,\cdots,g$.
\label{reported and IBNR process}
\end{Corollary}

As direct consequences, the univariate and multivariate marginal distributions of $\{N_1,N_2,\cdots\}$ can be described below. We denote the $k$-step transition probability matrix by $\boldsymbol{\Gamma}^k=\left(\gamma_{ij}(k)\right)_{g\times g}$, where $\gamma_{ij}(k)=P(C_{l+k}=j|C_l=i)$.

\begin{Corollary}
  For $l=1,2,\cdots$,~$N_l$ follows a mixed Pascal distribution with probability function
  \begin{equation}
    P(N_l=n)=\sum_{i=1}^g\pi_{li}p\left(n;m_i,(d_l-d_{l-1})\omega_l\theta\right).
  \end{equation}
  \label{marginal distribution}
\end{Corollary}

\begin{proof}
  Similar to the proof of Proposition \ref{Lambda}.
\end{proof}

\begin{Corollary}
  The $k$-variate joint distribution of $(N_{l_1},\ldots,N_{l_k})$ is
  \begin{align}
    \nonumber &P(N_{l_1}=n_{l_1},\ldots,N_{l_k}=n_{l_k})\\
    =&\sum_{i_1=1}^g\cdots\sum_{i_k=1}^g\beta_{(i_1,\dots,i_k)}\prod_{j=1}^kp(n_{l_j};m_{i_j},(d_{l_j}-d_{l_j-1})\omega_{l_j}\theta),
  \end{align}
  where
  \begin{equation}
    \beta_{(i_1,\cdots,i_k)}=\pi_{l_1,i_1}\gamma_{i_1,i_2}(l_2-l_1)\cdots\gamma_{i_{k-1},i_k}(l_k-l_{k-1}).
  \end{equation}
  \label{k-variate marginal}
\end{Corollary}

\begin{proof}
  We only prove the bivariate case and a similar argument can be applied to any higher dimension. For $l_1<l_2$, we have
  \begin{align*}
    &P(N_{l_1}=n_{l_1},N_{l_2}=n_{l_2})\\
    =&\sum_{i_1=1}^g\sum_{i_2=1}^gP(N_{l_1}=n_{l_1},N_{l_2}=n_{l_2},C_{l_1}=i_1,C_{l_2}=i_2)\\
    =&\sum_{i_1=1}^g\sum_{i_2=1}^gP(C_{l_1}=i_1,C_{l_2}=i_2)P(N_{l_1}=n_{l_1},N_{l_2}=n_{l_2}|C_{l_1}=i_1,C_{l_2}=i_2)\\
    =&\sum_{i_1=1}^g\sum_{i_2=1}^gP(C_{l_1}=i_1)P(C_{l_2}=i_2|C_{l_1}=i_1)P(N_{l_1}=n_{l_1}|C_{l_1}=i_1)P(N_{l_2}=n_{l_2}|C_{l_2}=i_2)\\
    =&\sum_{i_1=1}^g\sum_{i_2=1}^g\pi_{l_1,i_1}\gamma_{i_1,i_2}(l_2-l_1)\prod_{j=1}^2\left(p(n_{l_j};m_{i_j},(d_{l_j}-d_{l_j-1})\omega_{l_j}\theta\right)\\
    =&\sum_{i_1=1}^g\sum_{i_2=1}^g\beta_{(i_1,i_2)}\prod_{j=1}^2\left(p(n_{l_j};m_{i_j},(d_{l_j}-d_{l_j-1})\omega_{l_j}\theta\right).
     \qedhere
  \end{align*}
\end{proof}

Corollary \ref{k-variate marginal} shows that $k$-variate marginal of the discrete observations of the proposed model is a multivariate Pascal mixture, which is known to constitute a versatile class of discrete multivariate distributions (\citet{Badescu2015a}). Consequently the proposed Pascal-HMM provides great flexibility in modeling temporal dependence. The following theorem provides a closed-form expression for the ACF of $\{N_1,N_2,\cdots\}$.

\begin{Theorem}
  If we further assume that:
  \begin{enumerate}
    \item $(d_l-d_{l-1})\omega_l=1$, $l=1,2,\cdots$,
    \item the eigenvalues of $\boldsymbol{\Gamma}$ are all distinct and they are ordered as $1=e_1>e_2>\cdots>e_g\ge-1$,
  \end{enumerate}
  then the ACF (in the limiting sense) for $\{N_1,N_2,\ldots\}$ is given by
  \begin{equation}
    \rho(k)=\frac{Cov(N_l,N_{l+k})}{Var(N_l)}=\sum_{i=2}^gc_ie_i^k,~k=1,2,\cdots,
  \end{equation}
  where
  \begin{equation*}
    c_i=\frac{\boldsymbol{\delta}\boldsymbol{M}\boldsymbol{u}_i^T\boldsymbol{v}_i\boldsymbol{M}\boldsymbol{1}^T}{\sum_{i=1}^g\delta_im_i(m_i+\frac{1+\theta}{\theta})-(\sum_{i=1}^g\delta_im_i)^2}.
  \end{equation*}
  Here $\boldsymbol{M}=diag\{m_1,\ldots,m_g\}$, $\boldsymbol{1}=(1,\ldots,1)$, $\boldsymbol{v}_i$ and $\boldsymbol{u}_i^T$ are the left and right eigenvectors of $\boldsymbol{\Gamma}$ associated with $e_i$, and $\boldsymbol{v}_i\boldsymbol{u}_i^T=1$.
  \label{expression for rho(k)}
\end{Theorem}

\begin{Remark}
  The first assumption ensures that, in the limiting sense, $\{N_1,N_2,\ldots\}$ is a stationary time series, which means that the ACF is sufficient for describing its temporal dependence in the long run. Since one can adjust the exposure by scaling, this assumption essentially means that $(d_l-d_{l-1})\omega_l$ is irrelevant of $l$, which can be realized if we narrow/widen the lengths of periods when portfolio exposure is high/low. The second assumption is not restrictive in applications as matrices with distinct eigenvalues are dense in the matrix space.
\end{Remark}

\begin{proof}
  According to the Perron-Frobenius Theorem (e.g., see \citet{Bremaud2013}), we have the following decomposition
  \begin{equation*}
    \boldsymbol{\Gamma}^k=\sum_{i=1}^ge_i^k\boldsymbol{u}_i^T\boldsymbol{v}_i,
  \end{equation*}
  where $1=e_1>\cdots>e_g\ge-1$ are the distinct eigenvalues of $\boldsymbol{\Gamma}$ and $\boldsymbol{v}_i(\boldsymbol{u}_i^T)$ are the left(right) eigenvectors of $\boldsymbol{\Gamma}$ associated with $e_i$, satisfying $\boldsymbol{v}_i\boldsymbol{u}_i^T=1$. Since all the eigenvalues are distinct, all the left and right eigenvectors are determined up to multiplication by a non-zero scalar. In particular, we have $\boldsymbol{u}_1=c_1(1,\ldots,1)$ and $\boldsymbol{v}_1=c_2(\delta_1,\ldots,\delta_g)$, where $c_1c_2=1$ due to the constraint that $\boldsymbol{v}_i\boldsymbol{u}_i^T=1$.

  Similar to the proof of Corollary \ref{k-variate marginal}, one may check that
  \begin{align*}
    & Cov(N_l,N_{l+k})\\
    =&E(N_lN_{l+k})-E(N_l)^2\\
    =&\sum_{i=1}^g\sum_{j=1}^g \delta_i\gamma_{ij}(k)m_i\theta m_j\theta-\left(\sum_{i=1}^g\delta_im_i\theta\right)^2\\
    =&\theta^2\left(\boldsymbol{\delta}
    \begin{pmatrix}
      m_1 & \ldots & 0 \\
      \vdots & \ddots & \vdots \\
      0 & \ldots & m_g
    \end{pmatrix}
    \boldsymbol{\Gamma}^k
    \begin{pmatrix}
      m_1 \\
      \vdots \\
      m_g
    \end{pmatrix}-\left(\sum_{i=1}^g\delta_im_i\right)^2\right)\\
    =&\theta^2\left(\boldsymbol{\delta}
    \begin{pmatrix}
      m_1 & \ldots & 0 \\
      \vdots & \ddots & \vdots \\
      0 & \ldots & m_g
    \end{pmatrix}
    \left(\sum_{i=1}^ge_i^k\boldsymbol{u}_i^T\boldsymbol{v}_i\right)
    \begin{pmatrix}
      m_1 \\
      \vdots \\
      m_g
    \end{pmatrix}-\left(\sum_{i=1}^g\delta_im_i\right)^2\right)\\
    =& \theta^2\boldsymbol{\delta}
    \begin{pmatrix}
      m_1 & \ldots & 0 \\
      \vdots & \ddots & \vdots \\
      0 & \ldots & m_g
    \end{pmatrix}
    \boldsymbol{u}_1^T\boldsymbol{v}_1
    \begin{pmatrix}
      m_1 \\
      \vdots \\
      m_g
    \end{pmatrix}+
    \theta^2\sum_{i=2}^g\boldsymbol{\delta}
    \begin{pmatrix}
      m_1 & \ldots & 0 \\
      \vdots & \ddots & \vdots \\
      0 & \ldots & m_g
    \end{pmatrix}
    \boldsymbol{u}_i^T\boldsymbol{v}_i
    \begin{pmatrix}
      m_1 \\
      \vdots \\
      m_g
    \end{pmatrix}e_i^k\\
    &-\theta^2\left(\sum_{i=1}^g\delta_im_i\right)^2\\
    =&\theta^2\left(\left(\sum_{i=1}^g\delta_im_i\right)^2+
    \sum_{i=2}^g\boldsymbol{\delta}
    \begin{pmatrix}
      m_1 & \ldots & 0 \\
      \vdots & \ddots & \vdots \\
      0 & \ldots & m_g
    \end{pmatrix}
    \boldsymbol{u}_i^T\boldsymbol{v}_i
    \begin{pmatrix}
      m_1 \\
      \vdots \\
      m_g
    \end{pmatrix}e_i^k-\left(\sum_{i=1}^g\delta_im_i\right)^2\right)\\
    =&\theta^2\sum_{i=2}^g\boldsymbol{\delta}
    \begin{pmatrix}
      m_1 & \ldots & 0 \\
      \vdots & \ddots & \vdots \\
      0 & \ldots & m_g
    \end{pmatrix}
    \boldsymbol{u}_i^T\boldsymbol{v}_i
    \begin{pmatrix}
      m_1 \\
      \vdots \\
      m_g
    \end{pmatrix}e_i^k\\
    =&\theta^2 \sum_{i=2}^g\boldsymbol{\delta}\boldsymbol{M}\boldsymbol{u}_i^T\boldsymbol{v}_i\boldsymbol{M}\boldsymbol{1}^Te_i^k.
  \end{align*}
  The result follows by plugging
  \begin{equation*}
    Var(N)=\theta^2\left(\sum_{i=1}^g\delta_im_i\left(m_i+\frac{1+\theta}{\theta}\right)-(\sum_{i=1}^g\delta_im_i)^2\right)
  \end{equation*}
  into the expression for $\rho(k)$.
\end{proof}

Theorem \ref{expression for rho(k)} reveals that the ACF exhibits a power decaying pattern, which is similar to that of the popular ARIMA time series model. As a result, the proposed Pascal-HMM can achieve a wide range of temporal dependence structures by judiciously choosing $\boldsymbol{\Gamma}$. In particular, both positive and negative correlations can be realized, and this will be shown in more details in \citet{Badescu2015b}.

\section{\large An order statistics property} \label{section: An order statistics property}

When we compress the original data to its discretely observed process
$\{N_1,N_2,\cdots\}$, one may be concerned about any potential loss of information. We again note that here $N_l$ includes all the claims occurred during $[d_{l-1},d_l)$, no matter reported or not. In Theorem \ref{order statistics} below we show that a well-known order statistics property of the Poisson process holds in this more general case as well. This shows that the discretely observed  process preserves all the information about the claim arrival epochs. First we present an intermediate result in Lemma \ref{lemma for PP}, which can be easily checked using the independent increment property of the Poisson process.
\begin{Lemma}
  Assume that $\{N^a(t)\}$ is a marked Poisson process with intensity function $\lambda(t)$ and path-dependent mark $U$ with density function $p_{U|t}(u)$. The likelihood for the observations up to a given time $t$ is
\begin{equation*}
  P(N^a(t)=n,(T_i,U_i)\in(dt_i,du_i),i=1,\ldots,n)
	=e^{-\int_0^t\lambda(s)ds}\prod_{i=1}^n\left(\lambda(t_i)dt_i~p_{U_i|t_i}(du_i)\right).\qedhere
\end{equation*}
\label{lemma for PP}
\end{Lemma}

\begin{Theorem}
  For $l=1,\cdots, k$, we assume that there are $n_l$ claims that have occurred during the period $[d_{l-1},d_l)$, together with their marks ordered chronologically according to the arrival epochs $\{(T_i^{(l)},U_i^{(l)})~i=1,2,\ldots,n_l\}$. We then have
\begin{align}
  \nonumber &P\left((T_i^{(l)},U_i^{(l)})\in(dt_i^{(l)},du_i^{(l)}),l=1,\ldots,k,i=1,\ldots,n_l|N_l=n_l,l=1,\ldots,k\right)\\
=& \prod_{l=1}^k\left(n_l!\prod_{i=1}^{n_l}\left(\frac{dt_i^{(l)}}{d_l-d_{l-1}}p_{U_i^{(l)}}(du_i^{(l)})\right)\right).\label{order statistics property}
\end{align}
\label{order statistics}
In other words, Given the discrete observations, the joint distribution of the 
claim arrival epochs with markings are mutually independent and the epochs are uniformly distributed. 

\end{Theorem}

\begin{proof}
  Since when given $\Lambda(t)=\lambda(t)$, the proposed model $\{N^a(t)\}$ is a marked Poisson process with intensity function $\lambda(t)$, according to Lemma \ref{lemma for PP} the likelihood for the observations up to time $d_k$ is
\begin{align*}
  & P\left(N_l=n_l,(T_i^{(l)},U_i^{(l)})\in(dt_i^{(l)},du_i^{(l)}),l=1,\ldots,k,i=1,\ldots,n_l\right)\\
=& E_{\Lambda_1,\ldots,\Lambda_k}\left(\prod_{l=1}^k\left(e^{-(d_l-d_{l-1})\Lambda_l}\prod_{i=1}^{n_l}\left(\Lambda_ldt_i^{(l)}
p_{U_i^{(l)}}(du_i^{(l)})\right)\right)\right),
\end{align*}
where $E_{\Lambda_1,\ldots,\Lambda_k}$ signifies taking expectation with respect to $\Lambda_1,\ldots,\Lambda_k$. This can be further calculated as
\begin{align*}
  &\sum_{i_1=1}^g\cdots\sum_{i_k=1}^gP(C_1=i_1,\cdots,C_k=i_k)\\
& \cdot\int_0^\infty\cdots\int_0^\infty\left(\prod_{l=1}^k\left(e^{-(d_l-d_{l-1})\lambda_l}\prod_{i=1}^{n_l}\left(\lambda_ldt_i^{(l)}
p_{U_i^{(l)}}(du_i^{(l)})\right)\right)\right)\\
& \cdot f_{\Lambda_1,\cdots,\Lambda_k|C_1=i_1,\cdots,C_k=i_k}(\lambda_1,\ldots,\lambda_k)d\lambda_1\ldots d\lambda_k\\
=&\sum_{i_1=1}^g\cdots\sum_{i_k=1}^gP(C_1=i_1,\cdots,C_k=i_k)\\
& \cdot\prod_{l=1}^k\left(\int_0^{\infty}\lambda_l^{n_l}
e^{-(d_l-d_{l-1})\lambda_l}f_{\Lambda_l|C_l=i_l}(\lambda_l)d\lambda_l\right)\prod_{l=1}^k\prod_{i=1}^{n_l}\left(dt_i^{(l)}p_{U_i^{(l)}}(du_i^{(l)})\right).
\end{align*}
At the same time,
\begin{align*}
  & P(N_1=n_1,\cdots,N_k=n_k)\\
=&\sum_{i_1=1}^g\cdots\sum_{i_k=1}^gP(C_1=i_1,\cdots,C_k=i_k)\prod_{l=1}^kP(N_l=n_l|C_l=i_l)\\
=& \sum_{i_1=1}^g\cdots\sum_{i_k=1}^gP(C_1=i_1,\cdots,C_k=i_k)\\
&\cdot \prod_{l=1}^k\left(\int_0^\infty\frac{\left((d_l-d_{l-1})\lambda_l\right)^{n_l}e^{-(d_l-d_{l-1})\lambda_l}}{n_l!}f_{\Lambda_l|C_l=i_l}(\lambda_l)d\lambda_l\right).
\end{align*}
Combining these two equations yields the result.
\end{proof}

According to \eqref{order statistics property}, with the discrete observations are given, the joint distribution of the 
claim arrival epochs with markings is completely specified. Furthermore, this order statistics is critical in the estimation of the proposed model.

\section{\large The distribution of the numbers of reported and IBNR claims}
\label{section: Predicting the number of IBNR claims: a closed-form expression}

In this section we derive explicit analytical expressions for the distribution of the number of reported claims and the distribution of
the number of IBNR claims. 

\begin{Proposition}
Recall that $N_l^{r}$ is the number of reported claims in time interval $[d_{l-1},d_l)$.
Then the joint probability of $N_l^{r}, l=1, \cdots, k$, where $\tau=d_k$ is the valuation date, can explicitly be expressed 
as
\begin{equation}\label{jp_reported}
  P(N_1^{r}=n_1,\cdots,N_k^{r}=n_k)=\sum_{i_1=1}^g\cdots\sum_{i_k=1}^g\beta_{(i_1,\cdots,i_k)}\prod_{j=1}^kp(n_j;m_{i_j},\theta_j^{r}),
\end{equation}
where
\begin{equation*}
  \beta_{(i_1,\cdots,i_k)}=\pi_{1,i_1}\gamma_{i_1,i_2}\cdots\gamma_{i_{k-1},i_k}
\end{equation*}
and
\begin{equation*}
  \theta_j^{r}=\left(\int_{d_{j-1}}^{d_j}P_U(\tau-t)dt\right)\omega_j\theta,~j=1,2,\cdots,k.
\end{equation*}
Similarly, the  joint probability of $N_l^{IBNR}, l=1, \cdots, k$, the
numbers of IBNR claims in the same time intervals up to valuation date $\tau=d_k$ 
can explicitly be expressed as
\begin{equation}\label{jp_IBNR}
  P(N_1^{IBNR}=n_1,\cdots,N_k^{IBNR}=n_k)=\sum_{i_1=1}^g\cdots\sum_{i_k=1}^g\beta_{(i_1,\cdots,i_k)}\prod_{j=1}^kp(n_j;m_{i_j},\theta_j^{IBNR}),
\end{equation}
where
\begin{equation*}
  \theta_j^{IBNR}=\left(\int_{d_{j-1}}^{d_j}\left(1-P_U(\tau-t)\right)dt\right)\omega_j\theta,~j=1,2,\cdots,k.
\end{equation*}
\end{Proposition}
\begin{proof}
  It can be easily checked using Corollaries \ref{k-variate marginal} and \ref{reported and IBNR process}.
\end{proof}

In the following, we show that the total number of reported claims  and the total number of IBNR claims up to the valuation date
have a univariate Pascal mixture with a common scale parameter. 
For notational simplicity, we re-expressed the joint probabilities (\ref{jp_reported})  and (\ref{jp_IBNR}) as infinite series
 with a finite number of non-zero coefficients:
\begin{align}
  \nonumber &P(N_1^{r}=n_1,\cdots,N_k^{r}=n_k)\\
  =&\sum_{m_1=1}^\infty\cdots\sum_{m_k=1}^\infty\beta_{(m_1,\cdots,m_k)}\prod_{j=1}^kp(n_j;m_j,\theta_j^{r}),
  \label{infinite sum expression2}
\end{align}
and
\begin{align}
  \nonumber &P(N_1^{IBNR}=n_1,\cdots,N_k^{IBNR}=n_k)\\
  =&\sum_{m_1=1}^\infty\cdots\sum_{m_k=1}^\infty\beta_{(m_1,\cdots,m_k)}\prod_{j=1}^kp(n_j;m_j,\theta_j^{IBNR}),
  \label{infinite sum expression1}
\end{align}
where $\beta_{(m_1,\cdots,m_k)}=0$ if one of the integers $m_1, \cdots, m_k$ is not a shape parameter in the Pascal mixture (\ref{jp_reported}).
Denote 
$N^r=\sum_{l=1}^k N_l^r$, 
the total number of reported claims up to the valuation date $\tau$, and 
$N^{IBNR}=\sum_{l=1}^kN_l^{IBNR}$, the total number of IBNR claims.  Note that the both distributions are  multivariate Pascal mixtures with scale parameters varying over different dimensions.  To show  
the distribution of $N^r$ and $N^{IBNR}$ has a univariate Pascal mixture, we employ a two step procedure:
\begin{itemize}
  \item First re-express \eqref{infinite sum expression1} and \eqref{infinite sum expression2} using a common scale parameter over different dimensions. This is an approach presented in \citet{WillmotWoo2015}.
  \item Prove that the sum of the components of a multivariate Pascal mixture with a common scale parameter is a univariate Pascal mixture.
\end{itemize}

We begin with a lemma.
\begin{Lemma}
  Assume that $(N_1,\cdots,N_k)$ is a multivariate Pascal mixture with varying scale parameters over different dimensions, i.e.,
\begin{equation}
  P(N_1=n_1,\cdots,N_k=n_k)=\sum_{m_1=1}^\infty\cdots\sum_{m_k=1}^\infty\beta_{(m_1,\cdots,m_k)}\prod_{j=1}^kp(n_j;m_j,\theta_j).
\label{multivariate Pascal mixture with different scale parameters}
\end{equation}
Then its probability generating function (PGF) is given by
\begin{equation*}
  P(z_1,\cdots,z_k)=Q\left(\frac{1}{1-\theta_1(z_1-1)},\cdots,\frac{1}{1-\theta_k(z_k-1)}\right),
\end{equation*}
where
\begin{equation*}
  Q(z_1,\cdots,z_k)=\sum_{m_1=1}^\infty\cdots\sum_{m_k=1}^\infty\beta_{(m_1,\ldots,m_k)}\prod_{j=1}^kz_j^{m_j}. \qedhere
\end{equation*}
\label{PGF of multivariate Pascal mixture}
\end{Lemma}
\begin{proof}
  The proof is straightforward.
\end{proof}

\begin{Proposition}
  Assume that $(N_1,\cdots,N_k)$ is a multivariate Pascal mixture with varying scale parameters over different dimensions and its joint probability is given by \eqref{multivariate Pascal mixture with different scale parameters}. For any $\theta\le\min\{\theta_1,\cdots,\theta_k\}$, \eqref{multivariate Pascal mixture with different scale parameters} can be re-expressed as
\begin{equation}
  P(N_1=n_1,\cdots,N_k=n_k)=\sum_{m_1=1}^\infty\cdots\sum_{m_k=1}^\infty\tilde{\beta}_{(m_1,\cdots,m_k)}p(n_j;m_j,\theta),
  \label{multiPascalcommontheta}
\end{equation}
where
\begin{equation*}
  \tilde{\beta}_{(m_1,\cdots,m_k)}=\sum_{n_1=1}^{m_1}\cdots\sum_{n_k=1}^{m_k}\beta_{(n_1,\cdots,n_k)}\prod_{j=1}^k\dbinom{m_j-1}{n_j-1}\left(\frac{\theta}{\theta_j}\right)^{n_j}
\left(1-\frac{\theta}{\theta_j}\right)^{m_j-n_j}.
\end{equation*}
\label{change of scale parameter}
\end{Proposition}

\begin{proof}
  The proof is similar to that in \citet{WillmotWoo2015}. The PGF of $(N_1,\cdots,N_k)$ is
  \begin{align*}
    &P(z_1,\cdots,z_k)\\
    =&Q\left(\frac{1}{1-\theta_1(z_1-1)},\cdots,\frac{1}{1-\theta_k(z_k-1)}\right)\\
    =&Q\left(\frac{1}{1-\theta(z_1-1)}\frac{\theta}{\theta_1}\frac{1}{1-\frac{1-\frac{\theta}{\theta_1}}{1-\theta(z_1-1)}},\cdots,
    \frac{1}{1-\theta(z_k-1)}\frac{\theta}{\theta_k}\frac{1}{1-\frac{1-\frac{\theta}{\theta_k}}{1-\theta(z_k-1)}}\right)\\
    =&\tilde{Q}\left(\frac{1}{1-\theta(z_1-1)},\cdots,\frac{1}{1-\theta(z_k-1)}\right),
  \end{align*}
  where
  \begin{align*}
    &\tilde{Q}(z_1,\cdots,z_k)\\
    =&Q\left(\frac{\frac{\theta}{\theta_1}z_1}{1-(1-\frac{\theta}{\theta_1})z_1},\ldots,
    \frac{\frac{\theta}{\theta_k}z_k}{1-(1-\frac{\theta}{\theta_k})z_k}\right)\\
    =&\sum_{m_1\ge1}\cdots\sum_{m_k\ge1}\beta_{(m_1,\cdots,m_k)}\prod_{j=1}^k\left(\frac{\frac{\theta}{\theta_j}z_j}{1-(1-\frac{\theta}{\theta_j})z_j}\right)^{m_j}\\
    =&\sum_{m_1\ge1}\cdots\sum_{m_k\ge1}\beta_{(m_1,\cdots,m_k)}\left(\sum_{n_j\geq m_j}\dbinom{n_j-1}{m_j-1}\left(\frac{\theta}{\theta_j}\right)^{m_j}\left(1-\frac{\theta}{\theta_j}\right)^{n_j-m_j}z_j^{n_j}\right)\\
    =&\sum_{m_1\ge1}\cdots\sum_{m_k\ge1}\beta_{(m_1,\cdots,m_k)}\sum_{n_1\geq m_1}\cdots\sum_{n_k\geq m_k}\left(\prod_{j=1}^k\dbinom{n_j-1}{m_j-1}\left(\frac{\theta}{\theta_j}\right)^{m_j}\left(1-\frac{\theta}{\theta_j}\right)^{n_j-m_j}z_j^{n_j}\right)\\
    =&\sum_{n_1\geq 1}\cdots\sum_{n_k\geq 1}\left(\sum_{m_1=1}^{n_1}\cdots\sum_{m_k=1}^{n_k}\beta_{(m_1,\cdots,m_k)}\prod_{j=1}^k\dbinom{n_j-1}{m_j-1}\left(\frac{\theta}{\theta_j}\right)^{m_j}\left(1-\frac{\theta}{\theta_j}\right)^{n_j-m_j}\right)\prod_{j=1}^kz_j^{n_j}\\
    =&\sum_{n_1\geq 1}\cdots\sum_{n_k\geq 1}\tilde{\beta}_{(n_1,\cdots,n_k)}\prod_{j=1}^k z_j^{n_j}.
  \end{align*}
  The conclusion follows from Lemma \ref{PGF of multivariate Pascal mixture} and the uniqueness of the PGF.
\end{proof}

\begin{Lemma}
  If $(N_1,\cdots,N_k)$ is a multivariate Pascal mixture with common scale parameter given by \eqref{multiPascalcommontheta},
then each marginal $N_i$ has the following stochastic representation:
\begin{equation*}
  N_i=\sum_{j=1}^{M_i}G_{ij},~i=1,2,\cdots,k,
\end{equation*}
where $G_{ij}$ are i.i.d geometric random variables with mean $\theta$, and the count variables $(M_1,\cdots,M_k)$ have a joint probability function
\begin{equation*}
  P(M_1=m_1,\cdots,M_k=m_k)=\tilde{\beta}_{(m_1,\cdots,m_k)}.
\end{equation*}
\label{Stochastic repre}
\end{Lemma}
\begin{proof}
  It is easy to check by calculating the PGFs of both representations.
\end{proof}

\begin{Proposition}
  If $(N_1,\cdots,N_k)$ is a multivariate Pascal mixture with common scale parameter and its joint probability function given by \eqref{multiPascalcommontheta}, then $N=N_1+\ldots,N_k$ is a univariate Pascal mixture with the same scale parameter and its mixing weights are
\begin{equation*}
  \tilde{\beta}_m^N=\sum_{m_1+\cdots+m_k=m}\tilde{\beta}_{(m_1,\cdots,m_k)}.
\end{equation*}
\label{multivariate to univariate}
\end{Proposition}

\begin{proof}
  By Lemma \ref{Stochastic repre}, $N=N_1+\cdots+N_k=\sum_{i=1}^{M_1+\cdots+M_k}G_i$. As a result, $N$ is a univariate Pascal mixture whose mixing weight for the $i$-th component equals
\begin{align*}
  \tilde{\beta}_m^N=&P(M_1+\cdots+M_k=m)\\
=&\sum_{m_1+\cdots+m_k=m}P(M_1=m_1,\cdots,M_k=m_k)\\
=&\sum_{m_1+\cdots+m_k=m}\tilde{\beta}_{(m_1,\cdots,m_k)}.\qedhere
\end{align*}
\end{proof}

\begin{Theorem}
For any $0<\theta^r<\min\{\theta_1^r,\ldots,\theta_k^r\}$, the total number of reported claims up
 to the valuation date $\tau$ is a univariate Pascal mixture with probability function
\begin{equation*}
  P(N^{r}=n)=\sum_{m=1}^{\infty}\left(\sum_{m_1+\cdots+m_k=m}\tilde{\beta}_{(m_1,\cdots,m_k)}\right)p(n;m,\theta^{r}),
\end{equation*}
where
\begin{equation*}
  \tilde{\beta}_{(m_1,\cdots,m_k)}=\sum_{n_1=1}^{m_1}\cdots\sum_{n_k=1}^{m_k}\beta_{(n_1,\cdots,n_k)}\prod_{j=1}^k\dbinom{m_j-1}{n_j-1}
\left(\frac{\theta^{r}}{\theta_j^{r}}\right)^{n_j}\left(1-\frac{\theta^{r}}{\theta_j^{r}}\right)^{m_j-n_j}. \qedhere
\end{equation*}
Similarly,
for any $0<\theta^{IBNR}<\min\{\theta_1^{IBNR},\cdots,\theta_k^{IBNR}\}$, the total number of IBNR claims up to the
valuation date $\tau$ is an univariate Pascal mixture, with probability function
\begin{equation*}
  P(N^{IBNR}=n)=\sum_{m=1}^{\infty}\left(\sum_{m_1+\cdots+m_k=m}\tilde{\beta}_{(m_1,\cdots,m_k)}\right)p(n;m,\theta^{IBNR}),
\end{equation*}
where
\begin{equation*}
  \tilde{\beta}_{(m_1,\cdots,m_k)}=\sum_{n_1=1}^{m_1}\cdots\sum_{n_k=1}^{m_k}\beta_{(n_1,\cdots,n_k)}\prod_{j=1}^k\dbinom{m_j-1}{n_j-1}
\left(\frac{\theta^{IBNR}}{\theta_j^{IBNR}}\right)^{n_j}\left(1-\frac{\theta^{IBNR}}{\theta_j^{IBNR}}\right)^{m_j-n_j}. \qedhere
\end{equation*}
\label{closed-form for IBNR claim}
\end{Theorem}

\begin{proof}
   It is obvious by Propositions \ref{change of scale parameter} and \ref{multivariate to univariate}.
\end{proof}

Although Theorem \ref{closed-form for IBNR claim} shows that theoretically it has a closed-form expression under our model assumptions, there may be some computational issues. When one uses Proposition \ref{change of scale parameter} to unify the different scale parameters over the time intervals, the resulting multivariate Pascal mixture will have infinite number of terms. There is no guarantee that a truncation of this infinite series can give an adequate approximation. More importantly, the truncated terms might play an important role in deciding the tail shape of the predictive distribution. Due to these considerations, one might resort to the simulation technique when calculating the distribution in practice.

\section{\large Concluding remarks}\label{conclusion}

In this paper, we propose a marked Cox model for a portfolio claim arrival process along with its reporting delays. The model can take into consideration the exposure fluctuations and has a great versatility in modeling temporal dependence.
The model is mathematically tractable. We show that the associated reported claim process and IBNR claim process are also marked Cox processes with easily convertible intensity functions and marking distributions.
The model allows to produce an equivalent discretely observed process from the claim arrival process, the reported claim process and the IBNR process, and their joint distributions respectively. Furthermore, 
closed-form expressions are available for the ACF of the discretely observed processes. These properties are critically important from a data fitting and prediction perspective.

In \citet{Badescu2015b}, we will present an algorithm to fit the proposed model to data and to estimate all the model parameters
including the number of states and the transition probabilities of the Markov chain. The efficiency of the fitting algorithm and the versatility of the proposed model are illustrated through detailed simulation studies. The usefulness of the proposed model is also tested by applying it to a real insurance data set. We compare the predictive distribution of our model with
the over-dispersed Poisson model (ODP), one of the several stochastic models that underpin
the widely used chain-ladder method. The results show that our model can yield more accurate best estimates and a more realistic predictive distribution.

Our current work opens several potential research directions. One could introduce the time trend and the seasonal effect into the claim arrival process by incorporating time covariates either in the state-dependent distributions or in the transition probability matrix.
While we only include the reporting delay as the single marker in the marked Cox model, the model can easily be extended to the situation that has multiple markers, e.g., the multiple payments of a reported claim. They
can be potentially modeled as recurrent events or one can generalize the  model to be a marked Cox cluster process.
The temporal dependence structure in the proposed model can also be enhanced in at least two directions. One could replace the underlying Markov chain structure in the current model with a Markov process to make it a full generalization of the Markov-Modulated Poisson Process (MMPP). Another path is to impose a more complex temporal dependence structure, e.g., at both the latent process level and the observation level. This would be related to self-excited processes such as the Hawkes process.

\vskip 2 true cm

\noindent {\bf Acknowledgments}

\noindent
Andrei L. Badescu and X. Sheldon Lin are supported by grants from 
the Natural Sciences and Engineering Research Council of Canada (NSERC). The authors would like to thank Professor Ragnar Norberg and Professor Mario V. Wuthrich for sending their unpublished manuscripts.\\

\end{document}